\documentclass[12pt,twoside,reqno]{amsart}
\usepackage{pifont}
\usepackage{amsmath}
\usepackage{amsthm}
\usepackage{amsfonts}
\usepackage{amssymb}
\usepackage{latexsym}
\usepackage{amstext}
\usepackage{array}
\usepackage{graphicx}

\date{}
\allowdisplaybreaks[4] \footskip=15pt
\renewcommand{\uppercasenonmath}[1]{}

\newtheorem{Theorem}{Theorem}[section]

\newtheorem{Lemma}{Lemma}[section]
\newtheorem{Remark}{Remark}[section]

\newtheorem{Proposition}{Proposition}[section]

\theoremstyle{definition}
\newtheorem{Definition}{Definition}[section]

\usepackage{amscd,amsthm,amsfonts,amssymb,amsmath}
\usepackage[colorlinks=true]{hyperref}
\usepackage{fullpage}
\usepackage[all]{xy}

    \numberwithin{equation}{section}

\input txdtools
\begin{document}

\begin{center}

{\large  \bf A general construction of regular complete permutation polynomials}

\vskip 0.8cm
{\small Wei Lu$^1$ $\cdot$ Xia Wu$^1$ \footnote{Supported by NSFC (Nos. 11971102, 11801070, 11771007), the Fundamental Research Funds for the Central Universities.

  }} $\cdot$ Yufei Wang$^1$   $\cdot$ Xiwang Cao$^2$ \\

{\small $^1$School of Mathematics, Southeast University, Nanjing
210096, China}\\
{\small $^2$Department of Math, Nanjing University of Aeronautics and Astronautics, Nanjing 211100, China}\\
{\small E-mail:
 luwei1010@139.com, wuxiadd1980@163.com, 220211734@seu.edu.cn, xwcao@nuaa.edu.cn}\\
{\small $^*$Corresponding author. (Email: wuxiadd1980@163.com)}
\vskip 0.8cm
\end{center}

{\bf Abstract:} Let $r\geq 3$ be a positive integer and  $\mathbb{F}_q$ the finite field with $q$ elements. In this paper, we consider the $r$-regular  complete permutation property of maps with the form $f=\tau\circ\sigma_M\circ\tau^{-1}$ where $\tau$ is a PP over an extension field $\mathbb{F}_{q^d}$ and $\sigma_M$ is an invertible linear map  over  $\mathbb{F}_{q^d}$.
  We give a general construction of $r$-regular PPs for any positive integer $r$. When $\tau$ is  additive, we give a general construction of $r$-regular CPPs for any positive integer $r$. When $\tau$ is not additive, we give many examples of regular CPPs over the extension fields for $r=3,4,5,6,7$ and for arbitrary odd positive integer $r$.  These examples are the generalization of the first class of $r$-regular CPPs constructed by Xu,  Zeng and Zhang   (Des. Codes Cryptogr. 90, 545-575 (2022)).


\section{\bf Introduction}

 Let $p$ be a prime, $q$ be a power of $p$ and $\mathbb{F}_q$ the finite field with $q$ elements. A polynomial $f(x)\in{\mathbb{F}}_q[x]$ is called a \emph{permutation polynomial} (PP) over ${\mathbb{F}}_q$ if the associated polynomial function $f: c\mapsto f(c)$ from ${\mathbb{F}}_q$ into ${\mathbb{F}}_q$ is a permutation of ${\mathbb{F}}_q$. PPs play important roles in finite field theory and they have broad applications in coding theory, combinatorial designs, and cryptography   \cite{CCZ1998, DP2013, D2013, DQWYY2015, DY2006, D1999a, D1999b, H2015, LM1984, M1973}. The two most important concepts related to a permutation  is the existence of fixed points and the specification of its cycle structure.

  Let $f(x)$ be a PP over ${\mathbb{F}}_q$,  $e$  the identity map over ${\mathbb{F}}_q$ and   $n$  a positive integer. Let $f^{(n)}$ be   the $n$-th composite power of $f$, where the $n$-th composite power of $f$ is defined inductively by $f^{(n)}:=f\circ f^{(n-1)}=f^{(n-1)}\circ f$ and $f^{(1)}:=f,f^{(0)}:=e$, $f^{(-n)}:=(f^{-1})^{(n)}$.    The  polynomial $f(x)$ is called  an \emph{$n$-cycle permutation} if $f^{(n)}$ equals  the identity map $e$. On the one hand, PPs with long cycles (especially full cycles) can be used to generate key-stream sequences with large periods \cite{F1982, G1967, GG2005}. On the other hand,  PPs with short cycles (especially  involutions $f^{(2)}=e$) can be used   to construct Bent functions over finite fields \cite{CM2018,G1962,M2016},  to design codes \cite{G1962} and to  against  some cryptanalytic attacks \cite{CR2015}. In general, it is difficult to  determine the cycle structure of a PP. Very few known PPs whose explicit cycle structures have been obtained (see \cite{A1969, LM1991, RC2004, RMCC2008} for monomials and Dickson polynomials).

  A polynomial $f(x)\in{\mathbb{F}}_q[x]$ is called a \emph{complete permutation polynomial} (CPP) over ${\mathbb{F}}_q$ if  both $f(x)$ and $f(x)+x$ are permutations of  ${\mathbb{F}}_q$. These polynomials were introduced by Mann in the construction of orthogonal Latin squares \cite{M1942}. Niederreiter and Robinson later gave a detailed study of CPPs over finite fields \cite{NR1982}.  CPPs have widely applications in the design of nonlinear dynamic substitution device \cite{M1995,M1997},  the Lay-Massey scheme \cite{V1999}, the block cipher SMS4 \cite{DL2008}, the stream cipher Loiss \cite{FFZ2011}, the design of Hash functions \cite{SV1995, V1994}, quasigroups \cite{MM2009,MM2012a,MM2012b}, and  the constructions of some cryptographically strong functions \cite{MP2014,SG2012,ZHC2015}.

  A polynomial $f(x)\in{\mathbb{F}}_q[x]$ is called   a \emph{$r$-regular} PP over ${\mathbb{F}}_q$ if $f(x)$ is a PP over ${\mathbb{F}}_q$ with all the cycles of the same length $r$ (ignoring the fixed points). Regular PPs are very important in applications of turbo-like coding, low-density parity-check codes (LDPC) and block cipher designs \cite{B2003,RC2004,RMCC2008,SSP2012}. In \cite{MP2014}, a recursive construction of CPPs  over finite fields
via subfield functions was proposed to construct CPPs  with no fixed points  over   finite fields  with odd characteristic. The similar technique was used  to construct strong complete mappings  in \cite{M2021} and to construct some  $r$-regular CPPs over ${\mathbb{F}}_q$  with even characteristic for some small positive integers $r$ in \cite{XZZ2022}.

The main purpose of this paper is to generalize the technique used in \cite{MP2014,M2021,XZZ2022} and give a general construction of regular PPs and regular CPPs over extension fields. The maps considered in this paper are of the forms $f=\tau\circ\sigma_M\circ\tau^{-1}$ where $\tau$ is a PP over an extension field $\mathbb{F}_{q^d}$ and $\sigma_M$ is an invertible linear map  over  $\mathbb{F}_{q^d}$. By linear algebra, it is easy to determine the cycle structure of $f$ and then construct regular PPs. In order to get regular CPPs, the difficulty is to make sure that $f+e$ is also a PP over $\mathbb{F}_{q^d}$. When $\tau$ is additive, we have that
$$f+e=\tau\circ\sigma_M\circ\tau^{-1}+\tau\circ e\circ\tau^{-1}=\tau\circ(\sigma_M+e)\circ\tau^{-1}, $$
and then $f+e$ is a PP over $\mathbb{F}_{q^d}$ if and only if $\sigma_M+e$ is an invertible linear map over $\mathbb{F}_{q^d}$, which is very easy to construct. While when $\tau$ is not additive, in general,
$$f+e\neq \tau\circ(\sigma_M+e)\circ\tau^{-1}.$$
In section \ref{section 4}, based on the ideal used in the additive cases, we will give several examples of regular CPPs. For different example, we may need different method to prove $f+e$ is a PP.  We hope that we can find a general method to make sure $f+e$ is a PP over $\mathbb{F}_{q^d}$  in further research.

Comparing with the first class of $r$-regular CPPs constructed in \cite{XZZ2022}, our results have the following advantages. First, We rewrite the first class of $r$-regular CPPs construction in \cite{XZZ2022} in the new form  $f=\tau\circ\sigma_M\circ\tau^{-1}$, which make us can easily give a general construction of regular PPs and regular CPPs. Second, in our results $p$ can be an arbitrary prime while in \cite{XZZ2022} $p$ is equal to 2. Third, for a given prime number $r$, we can give many different constructions while in \cite{XZZ2022} there are only construction (except  when $r=7$, there are two constructions). Fourth, we construct $r$-regular CPPs for any positive  number $r$ while in \cite{XZZ2022}, only prime cases were considered.  Fifth, in most places, our proof is easier than theirs.

This paper is organized as follows. In section \ref{section Preliminaries}, we introduce some basic knowledge about cycle structure of permutation polynomials, cyclotomic polynomials over finite fields and some properties of linear maps in linear algebra. In section \ref{Section 3}, based on the properties of linear maps and cyclotomic polynomials, we give a construction method in three steps: the linear cases, the additive cases and the general cases. Every step is based on the previous step. In section \ref{section 4}, based on the construction method given in section \ref{Section 3}, we give many examples. Two of these examples include \cite[Theorem 1, 2]{XZZ2022} as special cases. In section \ref{section Concluding remarks}, we conclude this paper and give some problems for further research.

 \section{\bf Preliminaries}\label{section Preliminaries}

\subsection{Cycle structures of permutation polynomials}

 In this subsection, we prepare and discuss the  cycle structures of permutation polynomials.

   Let $q$ be a prime power and $\mathbb{F}_q$ the finite field with $q$ elements.


 \begin{Definition}\cite[Definition 4]{CMS2016}
 Let $f$ be a PP over ${\mathbb{F}}_q$ and $t$  a positive integer. A \emph{cycle} of $f$ is a subset ${x_1,\dots , x_t}$ of pairwise distinct elements of ${\mathbb{F}}_q$ such that $f(x_i)=x_{i+1}$ for $1\leq i\leq t-1$ and $f(x_t)=x_1$. The cardinality of a cycle is called its length.
 \end{Definition}

 The result below gives  some cycle structures of  $n$-cycle permutations.

 \begin{Proposition}\cite[Proposition 2.2]{CWZ2021}
Let $f$ be a PP over ${\mathbb{F}}_q$. Then f is an $n$-cycle permutation if and only if the length $l$ of each cycle of $f$ is no more than $n$ and $l\ |\ n$.
 \end{Proposition}

 Once we obtained an $n$-cycle permutation, it is natural to obtain more $n$-cycle permutations by composing itself in the following lemma.

 \begin{Lemma}\cite[Lemma 2.9]{CWZ2021}
 Assume that $f$ is an $n$-cycle permutation over ${\mathbb{F}}_q$. Then $f^{(k)}$ is also an $n$-cycle permutation, where $n \geq 3$, $1<k<n$.
 \end{Lemma}

 Furthermore,  we have the following result.

 \begin{Proposition}\cite[Proposition 2.11]{CWZ2021}\label{proposite 2.2}
Let $f$ and $g$ be  PPs over ${\mathbb{F}}_q$.  Furthermore, $f$ is an $n$-cycle permutation. Then $g\circ f\circ g^{-1}$ is also an $n$-cycle permutation. Moreover, $f$ and $g\circ f\circ g^{-1}$ have the same cycle structures.
 \end{Proposition}

 Now we recall the definitions of  regular PPs and  regular CPPs over ${\mathbb{F}}_q$.

\begin{Definition}\cite[Definition 1]{XZZ2022}
A PP $f$ over ${\mathbb{F}}_q$ is called  \emph{$r$-regular} if   all the cycles of $f$ have the same length $r$ (ignoring the fixed points).
\end{Definition}

\begin{Remark}\label{Remark 2.1}
 It is easy to see that,  if  $r$ is a prime, then  any non-identity  $r$-cycle permutation is $r$-regular.
\end{Remark}

\begin{Definition}\cite[Definition 2]{XZZ2022}
A CPP $f$ over ${\mathbb{F}}_q$ is called \emph{$r$-regular}   if   all the cycles of $f$ have the same length $r$ (ignoring the fixed points).
\end{Definition}

\begin{Remark}
 It is easy to see that,  when the characteristic of  ${\mathbb{F}}_q$  is even,   any CPP $f$ has only one fixed point.
\end{Remark}

\subsection{Maps from $\mathbb{F}_q$ to $\mathbb{F}_q$}

Let $\mathcal{A}:=\{a:\mathbb{F}_q\rightarrow\mathbb{F}_q\}$ be the set of all maps from $\mathbb{F}_q$ to $\mathbb{F}_q$. In this subsection, we define two operations on $\mathcal{A}$ and get a non-commutative ring $\mathcal{B}$.

First, we define the addition on $\mathcal{A}$ to be the addition of maps: for any $a, b\in \mathcal{A}$, the sum of  $a, b$ is $a+b\in \mathcal{A}$ such that
\begin{center}
$(a+b)(x):=a(x)+b(x),$ for any $x\in \mathbb{F}_q.$
\end{center}
It  is trivially verified that with respect to addition, $\mathcal{A}$ is a commutative group.

Next, we define the multiplication on $\mathcal{A}$ to be composition of maps: for any $a, b\in \mathcal{A}$, the product  of  $a, b$ is $a\circ b\in \mathcal{A}$ such that
\begin{center}
$(a\circ b)(x):=a(b(x)),$ for any $x\in \mathbb{F}_q.$
\end{center}
It  is trivially verified that the multiplication is associative, and has a unit element. Its unit element is of course the identity map: $e(x)=x$ for any $x\in \mathbb{F}_q.$

For all $a,b,c\in \mathcal{A}$, we have the \emph{right distributivity}
 \begin{equation*}
 (a+b)\circ c=a\circ c+b\circ c.
 \end{equation*}
But in general,
 \begin{equation*}
 c\circ (a+b)\neq c\circ a+c\circ b.
 \end{equation*}
So $(\mathcal{A},+,\circ)$ is not a ring. For the definition of a ring, one can see \cite[p.83]{L2002}.

In order to get a ring, we need to chose a suitable subset of $\mathcal{A}.$
\begin{Definition}
A map $a\in \mathcal{A}$ is called \emph{additive} if
\begin{equation}\label{additive}
a(x+y)=a(x)+a(y), {\rm{\ for\ any\ }} x,y\in \mathbb{F}_q.
\end{equation}
\end{Definition}

If  $c\in \mathcal{A}$ is additive, then for any $a, b\in \mathcal{A}$ and any $x\in \mathbb{F}_q,$ we have
\begin{align*}
(c\circ (a+b))(x)&=c((a+b)(x))=c(a(x)+b(x))\\
                 &=c(a(x))+c(b(x))=(c\circ a)(x)+(c\circ b)(x)\\
                 &=(c\circ a+c\circ b)(x).
\end{align*}
So we have the \emph{left distributivity}

\begin{equation}\label{left distributivity}
 c\circ (a+b)= c\circ a+c\circ b.
\end{equation}

Now, the following lemma is  trivial.
\begin{Lemma}
Let $\mathcal{B}:=\{a\in \mathcal{A}\  |\  a$ is  additive\}. Then $(\mathcal{B},+, \circ)$ is a non-commutative ring.
\end{Lemma}

\begin{Remark}
If $p$ is the characteristic of $\mathbb{F}_q$, it is  trivially verified that $a\in \mathcal{B}$ if and only if $a$ is  $\mathbb{F}_p$-linear.
\end{Remark}

The additive group  $(\mathbb{F}_q,+)$, together with the natural  operation of $\mathcal{B}$ on $\mathbb{F}_q$
\begin{center}
$ax:=a(x)$, for any $a\in \mathcal{B}$ and $x\in \mathbb{F}_q,$
\end{center}
is a left module over $\mathcal{B}.$

Let  $d$ be a positive integer and $V:=\mathbb{F}_q^d$   the column   vector space with dimension $d$ over $\mathbb{F}_q.$ Let  $M_{d\times d}(\mathcal{B})$ be the matrix ring over $\mathcal{B}$ and $M$  any $d\times d$ matrix in  $M_{d\times d}(\mathcal{B})$.  Then we can define  a \emph{linear map} $\sigma_M$ from $V$ to $V$ by the usual way
 \begin{equation}\label{linear map}
 \sigma_M(v)=Mv,{\ \rm{ for\  any \ }} v\in V.
 \end{equation}
Then the additive group  $(V,+)$ together with the above operation of $M_{d\times d}(\mathcal{B})$ on $V$ is a left module over $M_{d\times d}(\mathcal{B}).$

For any $y\in \mathbb{F}_q$, we can define an $a_y\in \mathcal{A}$ by $a_y(x):=yx$, for any $x\in \mathbb{F}_q.$ It is trivially verified that for any $y\in \mathbb{F}_q$, $a_y\in \mathcal{B}.$ So $\mathbb{F}_q$ can be consider as a commutative subring of $\mathcal{B}.$

\subsection{Cayley-Hamilton Theorem}
In this subsection, we consider  $V=\mathbb{F}_q^d$  as   a free module of dimension    $d$ over $\mathbb{F}_q.$



\begin{Definition}
Let $M$ be any $d\times d$ matrix over $\mathbb{F}_q$ and $\sigma_M$ the linear map associate with $M$. We define the \emph{characteristic polynomial} $P_{\sigma_M}(t):=P_M(t)$ to be the determinant
\begin{center}
$\mathrm{det}$$(tI_d-M)$
\end{center}
where $I_d$ is the unit $d\times d$ matrix. The characteristic polynomial is an element of $\mathbb{F}_q[t]$.

 \end{Definition}




The following result  is important in the proof of  our main theorem.

\begin{Theorem}\cite[p.561]{L2002}{(Cayley-Hamilton Theorem)}
Let $M$ be any $d\times d$ matrix over $\mathbb{F}_q.$ We have
\begin{center}
 $P_{\sigma_M}({\sigma_M})=0.$
\end{center}
\end{Theorem}

\subsection{Minimal polynomial } In this subsection, we recall a few concepts and facts from linear algebra \cite[p.60, 525]{R2003}.

If $\sigma$ is a linear map on the vector space $V=\mathbb{F}_q^d$, then a polynomial $h(t)\in \mathbb{F}_q[t]$ is said to \emph{annihilate} $\sigma$ if $h(\sigma)=0$, where 0 is the zero map on $V$. The uniquely determined monic polynomial of least positive degree with this property is called the \emph{minimal polynomial} of $\sigma$. Minimal polynomial  divides any other polynomial in $\mathbb{F}_q[x]$ annihilating $\sigma$. In particular, the minimal polynomial of $\sigma$ divides the characteristic polynomial $P_\sigma$ by Cayley-Hamilton Theorem.

A vector $v\in V$ is called a \emph{cyclic vector} for $\sigma$ if the vectors $\sigma^kv$, $k=0,1,2,\ldots$ span $V$. The following is a standard result from linear algebra.

\begin{Lemma}\cite[Lemma 2.34]{R2003}
Let $\sigma$ be a linear map on $V=\mathbb{F}_q^d$. Then $\sigma$ has a cyclic vector if and only if the characteristic and minimal polynomials for $\sigma$ are identical.
\end{Lemma}

For a monic polynomial
$$h(t)=t^k+h_{k-1}t^{k-1}+\cdots+h_1t+h_0$$
over $\mathbb{F}_q$, its \emph{companion matrix} M(h(t)) is given by
\begin{equation}
M(h(t))=\left( {\begin{array}{*{20}{c}}
0&1&0& \ldots &0\\
0&0&1& \ldots &0\\
0&0&0& \ldots &0\\
 \vdots & \vdots & \vdots & \ddots & \vdots \\
-h_0&-h_1&-h_2&\ldots &-h_{k - 1}
\end{array}} \right).
\end{equation}
Then $h(t)$ is the characteristic polynomial and the minimal polynomial of $M(h(t))$.

\subsection{Cyclotomic polynomials}
In this subsection, we recall some results about cyclotomic polynomials \cite[p.64]{R2003}.

\begin{Definition}\cite[Definition 2.44]{R2003}
Let $n$ be a positive integer not divisible by $p$, and $\zeta$ a primitive $n$-th root of unity over $\mathbb{F}_q$. Then the polynomial
$$Q_n(x)=\prod\limits_{s=1\atop\gcd (s,n) = 1}^{n} {(x - {\zeta ^s})}$$
is called the \emph{n}-th cyclotomic polynomial over $\mathbb{F}_q$.
\end{Definition}
The following results are basic.

\begin{Proposition}\cite[Theorem 2.47, Lemma 2.50]{R2003}
Let $n$ be a positive integer not divisible by $p$. Then
\begin{enumerate}
  \item $x^n-1=\prod\limits_{d|n}{Q_d(x)}$;
  \item if $d$ is a divisor of $n$ with $1\leq d<n$, then $Q_n(x)$ divides $(x^n-1)/(x^d-1)$.
\end{enumerate}
\end{Proposition}

\subsection{Univariate form and multivariable form} In this subsection, we discuss the univariate form and multivariable form of the same polynomial from the extension field  $\mathbb{F}_{q^d}$ to itself through the dual basis. For  other relations between univariate forms and multivariable forms, one can see \cite{MP2014,XZZ2022}.

It is well know that $\mathbb{F}_{q^d}$ is a vector space with dimension $d$ over $\mathbb{F}_q.$ The isomorphism between $\mathbb{F}_{q^d}$ and $\mathbb{F}_q^d$ through an fixed basis $\{\alpha_1,\dots,\alpha_d\}$ implies that each element $\mathbb{F}_{q^d}$ can be uniquely represented as
\begin{equation}\label{isomorphism}
x=\alpha_1x_1+\alpha_2x_2+\cdots+\alpha_dx_d,
\end{equation}
where $x_i\in \mathbb{F}_q.$
Let  $\{\beta_1,\dots,\beta_d\}$ be the dual basis of $\{\alpha_1,\dots,\alpha_d\}$, see \cite[p.58]{R2003}. Then  for  $1 \leq i, j \leq d,$  we have
\begin{equation*}
\operatorname{Tr}_{\mathbb{F}_{q^d} / \mathbb{F}_q}\left(\alpha_{i} \beta_{j}\right)=\left\{\begin{array}{ll}
0 & \text { for } i \neq j, \\
1 & \text { for } i=j.
\end{array}\right.
\end{equation*}
For $1 \leq j \leq d,$ we have
\begin{equation*}
  \operatorname{Tr}_{\mathbb{F}_{q^d} / \mathbb{F}_q}\left(x \beta_{j}\right)=\operatorname{Tr}_{\mathbb{F}_{q^d} / \mathbb{F}_q}\left(\alpha_{1} \beta_{j}\right)x_1+\cdots+\operatorname{Tr}_{\mathbb{F}_{q^d} / \mathbb{F}_q}\left(\alpha_{d} \beta_{j}\right)x_d=x_j.
\end{equation*}

On the one hand, if we have a univariate polynomial $f(x)\in \mathbb{F}_{q^d}[x]$, then we get $d$ multivariable polynomials $f_1(x_1,\dots,x_d),\dots, f_d(x_1,\dots,x_d)\in \mathbb{F}_q[x_1,\dots,x_d]$ such that
      \begin{center}
$\alpha_1f_1(x_1,\dots,x_d)+\cdots+\alpha_df_d(x_1,\dots,x_d)=f(x)$.
\end{center}

On the other hand, if we have $d$ multivariable polynomials $f_1(x_1,\dots,x_d),\dots, f_d(x_1,\dots,x_d)\in \mathbb{F}_q[x_1,\dots,x_d]$, then we can get  a univariate polynomial $f(x)\in \mathbb{F}_{q^d}[x]$ such that

\begin{center}
$f(x)=\alpha_1f_1(\operatorname{Tr}_{\mathbb{F}_{q^d} / \mathbb{F}_q}\left(x \beta_{1}\right),\dots,\operatorname{Tr}_{\mathbb{F}_{q^d} / \mathbb{F}_q}\left(x \beta_{d}\right))+ \cdots+\alpha_df_d(\operatorname{Tr}_{\mathbb{F}_{q^d} / \mathbb{F}_q}\left(x \beta_{1}\right),\dots,\operatorname{Tr}_{\mathbb{F}_{q^d} / \mathbb{F}_q}\left(x \beta_{d}\right))$.
\end{center}


\emph{In the rest of this paper, utilizing the isomorphism between $\mathbb{F}_{q^d}$ and $\mathbb{F}_q^d$ through an fixed basis $\{\alpha_1,\ldots,\alpha_d\},$ we will regard any element in $\mathbb{F}_{q^d}$ as a vector in $\mathbb{F}_q^d$ and any map $f$ from finite field $\mathbb{F}_{q^d}$ to $\mathbb{F}_{q^d}$ as a map from vector space $\mathbb{F}_q^d$ to $\mathbb{F}_q^d$. }

\section{\bf  constructions   of PPs and CPPs over the extension   fields}\label{Section 3}
In this section,   we studied the cycle structure and the regularity of PPs and CPPs over the extension   fields of $\mathbb{F}_q$ based on maps from $\mathbb{F}_q$ to $\mathbb{F}_q$. First, we studied the $\mathbb{F}_q$-linear cases in subsection \ref{linear subsection}. The main method is the Cayley-Hamilton Theorem. Next, based on the $\mathbb{F}_q$-linear cases and the composition of maps, we studied the additive cases in subsection \ref{additive subsection}.  Finally, based on the idea of method used in the additive cases, we studied the general cases in subsection \ref{general subsection}. The main difference between the  additive cases and the general cases is that additive maps have the left distributivity which do not have in general.

\subsection{\bf the $\mathbb{F}_q$-linear cases }\label{linear subsection}
In this subsection, we study the cycle structure and the regularity of linear maps. It is the foundation of next two subsections.  The following theorem   is based on linear algebra and the theory of characteristic polynomials.

\begin{Theorem}\label{linear theorem}
Assume that $h(t)\in \mathbb{F}_q[t]$ and $d=\mathrm{deg}(h(t))>0$. Let $M$ be any $d\times d$ matrix over $\mathbb{F}_q$  such that $P_M(t)=h(t)$. Let $\sigma_M$ be the linear map associate with $M$ from $\mathbb{F}_{q^d}$ to $\mathbb{F}_{q^d}$ and $e$ the identity map from $\mathbb{F}_{q^d}$ to $\mathbb{F}_{q^d}$. Then
\begin{enumerate}
  \item if $h(0)\neq 0$, then $\sigma_M$ is a PP over $\mathbb{F}_{q^d}$;
  \item if $h(t)\ |\  (t^n-1)$ for some positive integer $n$, then $\sigma_M$ is   an  $n$-cycle permutation over $\mathbb{F}_{q^d}$;
  \item if $h(-1)\neq 0$, then $\sigma_M+e$ is a PP over $\mathbb{F}_{q^d}$;
  \item if $h(t)\ |\  ((t+1)^m-1)$ for some positive integer $m$, then $\sigma_M+e$ is   an  $m$-cycle permutation over $\mathbb{F}_{q^d}$.
\end{enumerate}
\end{Theorem}
\begin{proof}
\begin{enumerate}
  \item Note that $0\neq h(0)=P_M(0)=\mathrm{det}$$(0I_d-M)=\mathrm{det}$$(-M)=(-1)^d\mathrm{det}(M)$. Hence $M$ is an invertible matrix and $\sigma_M$ is an invertible linear map.
  \item By the Cayley-Hamilton Theorem, we have $h(\sigma_M)=0$. If $h(t)| (t^n-1)$, then $\sigma_M^n-e=0$ and $\sigma_M^n=e$. So $\sigma_M$ is an  $n$-cycle permutation over $\mathbb{F}_{q^d}$.
  \item Note that the matrix of $\sigma_M+e$ is $M+I_d$ and its characteristic polynomial is $P_{\sigma_M+e}(t)=\mathrm{det}$$(tI_d-(M+I_d))=\mathrm{det}$$((t-1)I_d-M)=P_M(t-1)=h(t-1)$. Hence $(-1)^d\mathrm{det}(M+I_d)=P_{\sigma_M+e}(0)=h(-1)\neq 0$, $M+I_d$ is an invertible matrix and $\sigma_M+e$ is an invertible linear map.
  \item By the Cayley-Hamilton Theorem, $h(\sigma_M)=h((\sigma_M+e)-e)=P_{\sigma_M+e}(\sigma_M+e)=0$. Since $h(t)\ |\  ((t+1)^m-1)$, we have $(\sigma_M+e)^m-e=0$ and $(\sigma_M+e)^m=e$. So $\sigma_M+e$ is an  $m$-cycle permutation over $\mathbb{F}_{q^d}$.
\end{enumerate}
\qed\end{proof}

\begin{Remark}
For a given $h(t)$, there are many matrices $M$ satisfy that $P_M(t)=h(t)$ in Theorem \ref{linear theorem}.  For example, $M=M(h(t))$, the companion matrix of $h(t)$.
\end{Remark}

Now we   consider the regularity of some  CPPs constructed in Theorem \ref{linear theorem}. First, we consider the odd prime cases.
\begin{Proposition}
Let $r$ be an odd prime which is relatively prime to $p$. Assume that $h(t)\in \mathbb{F}_q[t]$ satisfies $h(t)\ |\ (t^r-1)$, $h(t)\neq t-1$ and $d=\mathrm{deg}(h(t))>0$. Let $M$ be any $d\times d$ matrix over $\mathbb{F}_q$  such that $P_M(t)=h(t)$. Let $\sigma_M$ be the linear map associate with $M$ from $\mathbb{F}_{q^d}$ to $\mathbb{F}_{q^d}$ and $e$ the identity map from $\mathbb{F}_{q^d}$ to $\mathbb{F}_{q^d}$. Then $\sigma_M$ is an $r$-regular CPP over $\mathbb{F}_{q^d}$.
\end{Proposition}
\begin{proof}
Since $h(t)|(t^r-1)$, by Theorem \ref{linear theorem}(2), $\sigma_M$ is   an  $r$-cycle permutation over $\mathbb{F}_{q^d}$. Since $h(t)\neq t-1$, $\sigma_M\neq e$. By Remark \ref{Remark 2.1},  $\sigma_M$ is $r$-regular.
Since $r$ is an odd prime and $h(t)\ |\ (t^r-1)$, we have  $h(-1)\neq 0$. By Theorem \ref{linear theorem}(3), $\sigma_M+e$ is a PP over $\mathbb{F}_{q^d}$. So $\sigma_M$ is an \emph{r}-regular CPP over $\mathbb{F}_{q^d}$.

\qed\end{proof}

Next, we we consider the composite number cases. Let $r$ be a composite number. We give two propositions. The CPPs in the first proposition are \emph{r}-regular while the CPPs in the second proposition are \emph{r}-cycle but not \emph{r}-regular.

\begin{Proposition}\label{proposition 3.2}
Let $r$ be a composite number which is relatively prime to $p$ and $Q_r(t)$ the $r$-th cyclotomic polynomial over $\mathbb{F}_q$. Assume that $h(t)\in \mathbb{F}_q[t]$ satisfies that $h(t)\ |\ Q_r(t)$ and  $d=\mathrm{deg}(h(t))>0$. Let $M$ be any $d\times d$ matrix over $\mathbb{F}_q$ such that $P_M(t)=h(t)$. Let $\sigma_M$ be the linear map associated with $M$ from $\mathbb{F}_{q^d}$ to $\mathbb{F}_{q^d}$ and $e$ the identity map from $\mathbb{F}_{q^d}$ to $\mathbb{F}_{q^d}$. Then $\sigma_M$ is an $r$-regular CPP over $\mathbb{F}_{q^d}$.
\end{Proposition}
\begin{proof}
Since $r$ is a composite number with $(r,p)=1$ and $h(t)\ |\ Q_r(t)$, we have $h(-1)\neq 0$ and $h(t)\ |\ (t^r-1)$. By $h(-1)\neq 0$ and Theorem \ref{linear theorem}(3), we have $\sigma_M+e$ is a PP over $\mathbb{F}_{q^d}$. By $h(t)\ |\ (t^r-1)$ and Theorem \ref{linear theorem}(2), we have $\sigma_M$ is an \emph{r}-cycle permutation over $\mathbb{F}_{q^d}$.

Next, we prove that $\sigma_M$ is $r$-regular. For any $x\in \mathbb{F}_{q^d}^*$, let $l$ be its length in  $\sigma_M$. That is to say, $x, {\sigma_M}x,\ldots,\sigma_M^{l-1}x$ are pairwise distinct and ${\sigma_M^l}x=x$. Let
$$W:={\mathbb{F}_q}x+{\mathbb{F}_q}{\sigma_M}x+\cdots+{\mathbb{F}_q}{\sigma_M^{l-1}}x.$$
Then $W$ is an invariant subspace of $\sigma_M$. Moreover, $h(t)$ and $t^l-1$ are two annihilating polynomials of $\sigma_M|_W$.
Let $g(t)$ be the minimal polynomial of $\sigma_M|_W$. Then
$$g(t)\ |\ (h(t),t^l-1).$$
By $h(t)\ |\ Q_r(t)$, we have $g(t)\ |\ (Q_r(t),t^l-1)$. If $l<r$, then by \cite[Lemma 2.50, p.66] {R2003}, $(Q_r(t),t^l-1)=1$ and $g(t)=1$, that is a contradiction. So $l=r$ and $\sigma_M$ is \emph{r}-regular.
Combining all, we have that $\sigma_M$ is an \emph{r}-regular CPP over $\mathbb{F}_{q^d}$.

\qed\end{proof}
\begin{Proposition}\label{proposition 3.3}
Let $r$ be a composite number which is relatively prime to $p$. Assume that $h(t)\in \mathbb{F}_q[t]$ satisfies that $h(t)\ |\ (t^r-1)$, $h(t)$ is reducible, $(h(t),\frac{t^r-1}{Q_r(t)})\neq 1$, $h(-1)\neq 0$ and $d=\mathrm{deg}(h(t))>0$. Let $M=M(h(t))$ be the companion matrix of $h(t)$. Let $\sigma_M$ be the linear map associated with $M$ from $\mathbb{F}_{q^d}$ to $\mathbb{F}_{q^d}$ and $e$ the identity map from $\mathbb{F}_{q^d}$ to $\mathbb{F}_{q^d}$. Then $\sigma_M$ is an r-cycle CPP over $\mathbb{F}_{q^d}$ but $\sigma_M$ is not r-regular.
\end{Proposition}
\begin{proof}
By $h(t)\ |\ (t^r-1)$, $h(-1)\neq 0$ and Theorem \ref{linear theorem}(1)(2)(3), we have $\sigma_M$ is an \emph{r}-cycle CPP over $\mathbb{F}_{q^d}$.

Next, we prove that $\sigma_M$ is not \emph{r}-regular. Since $h(t)$ is reducible and $(h(t),\frac{t^r-1}{Q_r(t)})\neq 1$, there exists a proper divisor of $r$, i.e. $l\ |\ r$ such that $(h(t),Q_l(t))\neq 1$. Let $h_1(t)$ be an irreducible factor of $(h(t),Q_l(t))$ over $\mathbb{F}_q$ and $h_2(t)=\frac{h(t)}{h_1(t)}$.

Since $M=M(h(t))$ is the companion matrix of $h(t)$, we have that $h(t)$ is the minimal polynomial of $M$. Since $h_2(t)$ is a proper factor of $h(t)$, we have $h_2(\sigma_M)\neq 0$.
Let $0\neq x_0\in \mathbb{F}_{q^d}$ such that $x_1=h_2(\sigma_M)(x_0)\neq 0$. Note that $h_1(\sigma_M)$ and $h_2(\sigma_M)$ are linear maps. So
 $$h_1(\sigma_M)(x_1)=h_1(\sigma_M)\circ h_2(\sigma_M)(x_0)=h_1(M)h_2(M)x_)=h(M)x_0=0.$$

Since $h_1(t)\ |\ Q_l(t)$, we have $h_1(t)\ |\ (t^l-1)$ and $(\sigma_M^l-e)(x_1)=0$. So $\sigma_M^l(x_1)=x_1$ and $\sigma_M$ is not \emph{r}-regular.

\qed\end{proof}

\begin{Remark}
For a given $h(t)$, there are many matrices $M$ satisfy that $P_M(t)=h(t)$ in Proposition \ref{proposition 3.2}. But in Proposition \ref{proposition 3.3}, in order to make sure that $h(t)$ is the minimal polynomial of $M$, we choose $M=M(h(t))$, the companion matrix of $h(t)$.
\end{Remark}

\subsection{\bf the additive cases }\label{additive subsection}
In this subsection, we study the cycle structure and the regularity of additive maps. It is based on the $\mathbb{F}_q$-linear cases and the composition of maps.  The main advantage of additive maps is that they have the left distributivity. The following theorem is the main result of this subsection.

\begin{Theorem}\label{Theorem 3.2}
Assume that $h(t)\in \mathbb{F}_q[t]$ and $d=\mathrm{deg}(h(t))>0$. Let $M$ be any $d\times d$ matrix over $\mathbb{F}_q$ such that $P_M(t)=h(t)$. Let $\tau_1,\tau_2$ be any additive PPs over $\mathbb{F}_{q^d}$. Assume that $\sigma=\tau_1\circ \sigma_M\circ \tau_2$. Then
\begin{enumerate}
  \item if $h(0)\neq 0$, then $\sigma$ is a PP over $\mathbb{F}_{q^d}$;
  \item if $\tau_1\circ\tau_2=e$ and $h(t)\ |\ (t^n-1)$ for some positive integer $n$, then $\sigma$ is   an  $n$-cycle permutation over $\mathbb{F}_{q^d}$;
  \item if $\tau_1\circ\tau_2=e$ and $h(-1)\neq 0$, then $\sigma+e$ is a PP over $\mathbb{F}_{q^d}$;
  \item if $\tau_1\circ\tau_2=e$ and $h(t)\ |\ ((t+1)^m-1)$ for some positive integer $m$, then $\sigma+e$ is   an  $m$-cycle permutation over $\mathbb{F}_{q^d}$.
\end{enumerate}
\end{Theorem}
\begin{proof}
\begin{enumerate}
  \item By Theorem \ref{linear theorem}(1), $\sigma_M$ is a PP over $\mathbb{F}_{q^d}$. Combining that $\tau_1$ and $\tau_2$ are PPs over $\mathbb{F}_{q^d}$, we have $\sigma=\tau_1\circ\sigma_M\circ\tau_2$ is a PP over $\mathbb{F}_{q^d}$.
  \item By  Theorem \ref{linear theorem}(2), $\sigma_M$ is an  $n$-cycle permutation over $\mathbb{F}_{q^d}$. By Proposition \ref{proposite 2.2}, we have $\sigma=\tau_1\circ\sigma_M\circ\tau_1^{-1}$ is an $n$-cycle permutation over $\mathbb{F}_{q^d}$.
  \item By Theorem \ref{linear theorem}(3), $\sigma_M+e$ is a PP over $\mathbb{F}_{q^d}$. Since $\tau_1$ is additive, we have
      $$\tau_1\circ(\sigma_M+e)\circ\tau_1^{-1}=(\tau_1\circ\sigma_M+\tau_1)\circ\tau_1^{-1}
      =\tau_1\circ\sigma_M\circ\tau_1^{-1}+\tau_1\circ\tau_1^{-1}
      =\sigma+e$$
      So $\sigma+e=\tau_1\circ(\sigma_M+e)\circ\tau_1^{-1}$ is a PP over $\mathbb{F}_{q^d}$.
  \item By Theorem \ref{linear theorem}(4) and Proposition \ref{proposite 2.2}, $\sigma+e=\tau_1\circ(\sigma_M+e)\circ\tau_1^{-1}$ is an  $m$-cycle permutation over $\mathbb{F}_{q^d}$.
\end{enumerate}
\qed\end{proof}

When $\tau_1$ is additive, we always have
$$\sigma+e=\tau_1\circ(\sigma_M+e)\circ\tau_1^{-1}.$$

Then we have
\begin{enumerate}
  \item $\sigma+e$ is a PP over $\mathbb{F}_{q^d}$ if and only if $\sigma_M+e$ is an invertible linear map over $\mathbb{F}_{q^d}$;
  \item if $\sigma+e$ is a PP over $\mathbb{F}_{q^d}$, then $\sigma+e$ has the same cycle structure with $\sigma_M+e$.
\end{enumerate}

Now the following propositions about the regularity of some CPPs is easy to get by the propositions in subsection \ref{linear subsection}.

\begin{Proposition}\label{proposition 3.4}
Let $r$ be an odd prime which is relatively prime to $p$. Assume that $h(t)\in \mathbb{F}_q[t]$ satisfies that $h(t)\ |\ (t^r-1)$, $h(t)\neq t-1$ and $d=\mathrm{deg}(h(t))>0$. Let $M$ be any $d\times d$ matrix over $\mathbb{F}_q$ such that $P_M(t)=h(t)$. Let $\tau_1$ be any additive PP over $\mathbb{F}_{q^d}$ and $\sigma=\tau_1\circ\sigma_M\circ\tau_1^{-1}$. Then $\sigma$ is an r-regular CPP over $\mathbb{F}_{q^d}$.
\end{Proposition}

\begin{Proposition}\label{proposition 3.5}
Let $r$ be a composite number which is relatively prime to $p$ and $Q_r(t)$ the $r$-th cyclotomic polynomial over $\mathbb{F}_q$. Assume that $h(t)\in \mathbb{F}_q[t]$ satisfies that $h(t)\ |\ Q_r(t)$ and $d=\mathrm{deg}(h(t))>0$. Let $M$ be any $d\times d$ matrix over $\mathbb{F}_q$ such that $P_M(t)=h(t)$. Let $\tau_1$ be any additive PP over $\mathbb{F}_{q^d}$ and $\sigma=\tau_1\circ\sigma_M\circ\tau_1^{-1}$. Then $\sigma$ is an r-regular CPP over $\mathbb{F}_{q^d}$.
\end{Proposition}

\begin{Proposition}\label{proposition 3.6}
Let $r$ be a composite number which is relatively prime to $p$. Assume that $h(t)\in \mathbb{F}_q[t]$ satisfies that $h(t)\ |\ t^r-1$, $h(t)$ is reducible, $(h(t),\frac{t^r-1}{Q_r(t)})\neq 1$, $h(-1)\neq 0$ and $d=\mathrm{deg}(h(t))>0$. Let $M=M(h(t))$ be the companion matrix of $h(t)$. Let $\tau_1$ be any additive PP over $\mathbb{F}_{q^d}$ and $\sigma=\tau_1\circ\sigma_M\circ\tau_1^{-1}$. Then $\sigma$ is an r-cycle CPP over $\mathbb{F}_{q^d}$ but $\sigma$ is not r-regular.
\end{Proposition}

\subsection{\bf the general cases }\label{general  subsection}
Let $\tau_1$ be any PP over $\mathbb{F}_{q^d}$ and $\sigma=\tau_1\circ\sigma_M\circ\tau_1^{-1}$. Then
\begin{enumerate}
  \item $\sigma$ is a PP over $\mathbb{F}_{q^d}$ if and only if $\sigma_M$ is  an invertible linear map over $\mathbb{F}_{q^d}$;
  \item if $\sigma$ is a PP over $\mathbb{F}_{q^d}$, then $\sigma$ has the same cycle structure with $\sigma_M$.
\end{enumerate}
So it is easy to construct regular PPs over $\mathbb{F}_{q^d}$ by subsection \ref{additive subsection}.

When we want to construct regular CPPs over $\mathbb{F}_{q^d}$, the situation is complicated.
When $\tau_1$ is additive,  $\sigma+e=\tau_1\circ(\sigma_M+e)\circ\tau_1^{-1}$, and it is easy   by subsection \ref{additive subsection}.
But when $\tau_1$ is not additive, $\sigma+e\neq\tau_1\circ(\sigma_M+e)\circ\tau_1^{-1}$ in general, and it is difficult to know when $\sigma+e$ is a PP over $\mathbb{F}_{q^d}$. We only have the following results.

\begin{Theorem}\label{Theorem 3.3}
Assume that $h(t)\in \mathbb{F}_q[t]$ and $d=\mathrm{deg}(h(t))>0$. Let $M$ be any $d\times d$ matrix over $\mathbb{F}_q$ such that $P_M(t)=h(t)$. Let $\tau_1,\tau_2$ be any PPs over $\mathbb{F}_{q^d}$. Assume that $\sigma=\tau_1\circ\sigma_M\circ\tau_2$. Then
\begin{enumerate}
  \item if $h(0)\neq 0$, then $\sigma$ is a PP over $\mathbb{F}_{q^d}$;
  \item if $\tau_1\circ\tau_2=e$ and $h(t)\ |\ (t^n-1)$ for some positive integer $n$, then $\sigma$ is an n-cycle permutation over $\mathbb{F}_{q^d}$.
\end{enumerate}
\end{Theorem}

\begin{Remark}
When $\tau_1$ is not additive, it is difficult to know when $\sigma+e$ is a PP over $\mathbb{F}_{q^d}$. In the next section, we will give many CPPs over $\mathbb{F}_{q^d}$. In each example, we must prove that $\sigma+e$ is a PP over $\mathbb{F}_{q^d}$ carefully.
\end{Remark}

In general, regular CPPs are not easy to construct. Bur regular PPs are easily constructed from Theorem \ref{Theorem 3.3}.

\begin{Proposition}\label{proposition 3.7}
Let $r$ be an odd prime which is relatively prime to $p$. Assume that $h(t)\in \mathbb{F}_q[t]$ satisfies that $h(t)\ |\ (t^r-1)$, $h(t)\neq t-1$ and $d=\mathrm{deg}(h(t))>0$. Let $M$ be any $d\times d$ matrix over $\mathbb{F}_q$ such that $P_M(t)=h(t)$. Let $\tau_1$ be any PP over $\mathbb{F}_{q^d}$ and $\sigma=\tau_1\circ\sigma_M\circ\tau_1^{-1}$. Then $\sigma$ is an r-regular PP over $\mathbb{F}_{q^d}$.
\end{Proposition}

\begin{Proposition}\label{proposition 3.8}
Let $r$ be a composite number which is relatively prime to $p$ and $Q_r(t)$ the r-th cyclotomic polynomial over $\mathbb{F}_q$. Assume that $h(t)\in \mathbb{F}_q[t]$ satisfies that $h(t)\ |\ Q_r(t)$ and $d=\mathrm{deg}(h(t))>0$. Let $M$ be any $d\times d$ matrix over $\mathbb{F}_q$ such that $P_M(t)=h(t)$. Let $\tau_1$ be any PP over $\mathbb{F}_{q^d}$ and $\sigma=\tau_1\circ\sigma_M\circ\tau_1^{-1}$. Then $\sigma$ is an r-regular PP over $\mathbb{F}_{q^d}$.
\end{Proposition}

\begin{Proposition}\label{proposition 3.9}
Let $r$ be a composite number which is relatively prime to $p$. Assume that $h(t)\in \mathbb{F}_q[t]$ satisfies that $h(t)\ |\ t^r-1$, $h(t)$ is reducible, $(h(t),\frac{t^r-1}{Q_r(t)})\neq 1$, $h(-1)\neq 0$ and $d=\mathrm{deg}(h(t))>0$. Let $M=M(h(t))$ be the companion matrix of $h(t)$. Let $\tau_1$ be any  PP over $\mathbb{F}_{q^d}$ and $\sigma=\tau_1\circ\sigma_M\circ\tau_1^{-1}$. Then $\sigma$ is an r-cycle PP over $\mathbb{F}_{q^d}$ but $\sigma$ is not r-regular.
\end{Proposition}

\section{ \bf examples of regular CPPs over the extension fields}\label{section 4}

In the section, we give several constructions of regular CPPs based on the results in Section \ref{Section 3}.  These examples show the power of our method in constructing regular CPPs.

In order to make sure that $f$ is a $r$-regular CPP over $\mathbb{F}_{q^d}$, we need to prove that:
\begin{enumerate}
  \item $f$ is a PP over $\mathbb{F}_{q^d}$;
  \item $f$ is $r$-regular over $\mathbb{F}_{q^d}$;
  \item $f+e$ is a PP over $\mathbb{F}_{q^d}$.
\end{enumerate}

When $f$ has the form $f=\tau\circ \sigma_M\circ\tau^{-1}$, by Section \ref{Section 3}, for suitable $M$, it is easy to make $f$ to be a $r$-regular PP over $\mathbb{F}_{q^d}$. When $\tau$ is additive, it is also easy to choose suitable $M$ to make $f+e$ to be a PP over $\mathbb{F}_{q^d}$ and then $f$ is a $r$-regular CPP over $\mathbb{F}_{q^d}$. But when $\tau$ is not additive, it is not easy to get $r$-regular CPPs. One can see the following examples.

In the rest of this section, we will consider the $r$-regular complete permutation property of some maps over the extension fields for $r=3,4,5,6,7$ and for arbitrary odd positive integer $r$.

\subsection{$r=3$}
In this subsection, we present constructions of 3-regular CPPs.

\begin{Proposition}\label{proposition 4.1}
Assume that $p\neq3$. Let $h(t)=Q_3(t)=\frac{t^3-1}{t-1}=t^2+t+1\in \mathbb{F}_q[t]$ and $M\in M_{2\times 2}(\mathbb{F}_q)$ such that $P_M(t)=h(t)$. Let $a_1, a_2$ be any PPs over $\mathbb{F}_q$ and $\tau((x_1,x_2)^T)=(a_1(x_1),a_2(x_2))^T$ for any $x_1,x_2\in \mathbb{F}_q$. Let $\sigma=\tau\circ\sigma_M\circ\tau^{-1}$.
Then $\sigma$ is a $3$-regular PP over $\mathbb{F}_{q^2}$ by Proposition \ref{proposition 3.7}. Moveover,
\begin{enumerate}
  \item if $a_1,a_2$ are additive, then $\sigma$ is a $3$-regular CPP over $\mathbb{F}_{q^2}$;
  \item if $a_1,a_2$ are additive and $p=2$, then  $\sigma+e$ is also a $3$-regular CPPs over $\mathbb{F}_{q^2}$;
  \item if $a_2=e$ and
  \begin{equation*}
  M=\left(
  \begin{array}{ccc}
    0 & m\\
    -m^{-1} & -1\\
    \end{array}
  \right)
  \end{equation*}
  where $m\in \mathbb{F}_q^*$, then $\sigma$ is a $3$-regular CPP over $\mathbb{F}_{q^2}$; moreover, if $p=2$, then $\sigma+e$ is also a $3$-regular CPPs over $\mathbb{F}_{q^2}$;
  \item if $a_2=e$ and
    \begin{equation*}
  M=\left(
  \begin{array}{ccc}
    -1 & 1\\
    -1 & 0\\
    \end{array}
  \right)
  \end{equation*}
  then $\sigma$ is a $3$-regular CPP over $\mathbb{F}_{q^2}$.
\end{enumerate}
\end{Proposition}
\begin{proof}
\begin{enumerate}
  \item We only need to prove that $\sigma+e$ is a permutation of $\mathbb{F}_{q^2}$. Since $a_1,a_2$ are additive and $h(-1)\neq 0$, by Theorem \ref{Theorem 3.2}(3), we have $\sigma+e$ is a PP over $\mathbb{F}_{q^2}$.
  \item Since $p=2$, we have $(t+1)^3-1=t^3+3t^2+3t=t^3+t^2+t$ and then $h(t)=t^2+t+1\ |\ ((t+1)^3-1)$. By Theorem \ref{Theorem 3.2}(4), $\sigma+e$ is a 3-regular PP over $\mathbb{F}_{q^2}$. Meanwhile, $(\sigma+e)+e=\sigma$ is a PP over $\mathbb{F}_{q^2}$. So $\sigma+e$ is  also a 3-regular CPP over $\mathbb{F}_{q^2}$.
  \item Let
  \begin{equation}
  \begin{aligned}
  \tau_1\begin{pmatrix}
  x_1\\
  x_2\\
  \end{pmatrix}
  =
  \begin{pmatrix}
  x_1\\
  -m^{-1}a_1^{-1}(x_2)\\
  \end{pmatrix}
  ,
  \tau_2\begin{pmatrix}
  x_1\\
  x_2\\
  \end{pmatrix}
  =
  \begin{pmatrix}
  x_1\\
  a_1(mx_2)\\
  \end{pmatrix}
  \end{aligned}
  \end{equation}
  and
    \begin{equation*}
  M_1=\left(
  \begin{array}{ccc}
    1 & 1\\
    1 & 0\\
    \end{array}
  \right).
  \end{equation*}
  It is easy that $\tau_1,\tau_2,\sigma_{M_1}$ are PPs over $\mathbb{F}_{q^2}$.
  By calculating,
  \begin{equation}
  \begin{aligned}
  (\sigma+e)\begin{pmatrix}
  x_1\\
  x_2\\
  \end{pmatrix}
  =
  \begin{pmatrix}
  a_1(mx_2)+x_1\\
  -m^{-1}a_1^{-1}(x_1)\\
  \end{pmatrix}
  =
  \tau_1\circ\sigma_{M_1}\circ\tau_2\begin{pmatrix}
  x_1\\
  x_2\\
  \end{pmatrix}
  \end{aligned}.
  \end{equation}
   So $\sigma+e=\tau_1\circ\sigma_{M_1}\circ\tau_2$ is a PP over $\mathbb{F}_{q^2}$;\label{3}

  Moreover, when $p=2$, we have $\tau_1\circ\tau_2=e$ and $P_{M_1}(t)=t^2+t+1$. By  Theorem \ref{Theorem 3.3}(2), we have $\sigma+e=\tau_1\circ\sigma_{M_1}\circ\tau_1^{-1}$ is a 3-regular CPP over $\mathbb{F}_{q^2}$. Meanwhile, $(\sigma+e)+e=\sigma$ is a PP over $\mathbb{F}_{q^2}$. So $\sigma+e$ is    also  a 3-regular CPP over $\mathbb{F}_{q^2}$.
  \item Let
  \begin{equation}
  \begin{aligned}
  \tau_1\begin{pmatrix}
  x_1\\
  x_2\\
  \end{pmatrix}
  =
  \begin{pmatrix}
  a_1(x_1)+a_1(x_2)\\
  x_1\\
  \end{pmatrix}
  ,
  \tau_2\begin{pmatrix}
  x_1\\
  x_2\\
  \end{pmatrix}
  =
  \begin{pmatrix}
   x_2\\
  a_1^{-1}(x_1)-x_2\\
  \end{pmatrix}
  \end{aligned}
  \end{equation}
   and
   \begin{equation*}
  M_1=\left(
  \begin{array}{ccc}
    0 & -1\\
    1 & 1\\
    \end{array}
  \right).
  \end{equation*}
  It is easy that $\tau_1,\tau_2,\sigma_{M_1}$ are PPs over $\mathbb{F}_{q^2}$. By calculating, we have
  \begin{equation}
  \begin{aligned}
  (\sigma+e)\begin{pmatrix}
  x_1\\
  x_2\\
  \end{pmatrix}
  =
  \begin{pmatrix}
  a_1(-a_1^{-1}(x_1)+x_2)+x_1\\
  -a_1^{-1}(x_1)+x_2\\
  \end{pmatrix}
  =
  \tau_1\circ\sigma_{M_1}\circ\tau_2\begin{pmatrix}
  x_1\\
  x_2\\
  \end{pmatrix}
  \end{aligned}.
  \end{equation}
  So $\sigma+e=\tau_1\circ\sigma_{M_1}\circ\tau_2$ is a PP over $\mathbb{F}_{q^2}$.
\end{enumerate}
\qed\end{proof}

\begin{Remark}
\begin{enumerate}
  \item In the Proposition \ref{proposition 4.1}(3), if $m=1$, then $M=M(h(t))$ and it becomes \cite[Corollary 1]{XZZ2022}. Our proof is different from theirs and is much easier than theirs.
  \item If $a_1=e$ and
  \begin{equation*}
  M=\left(
  \begin{array}{ccc}
    0 & m\\
    -m^{-1} & -1\\
    \end{array}
  \right),
  \end{equation*}
  we can also prove that $\sigma$ is a $3$-regular CPP over $\mathbb{F}_{q^2}$.
  \item In general, if none of $a_1,a_2$ are equal to $e$, it is difficult to know whether $\sigma+e$ is a PP over $\mathbb{F}_{q^2}$.
  \item In Proposition \ref{proposition 4.1}(4), in general, $\tau_1\circ\tau_2\neq e$. So we can not use Theorem \ref{Theorem 3.3}(2) to get the cycle structure of $\sigma+e=\tau_1\circ\sigma_{M_1}\circ\tau_2$.
\end{enumerate}
\end{Remark}

\subsection{$r=4$}
In this subsection, we present constructions of 4-regular CPPs.

\begin{Proposition}\label{proposition 4.2}
Assume that $p\neq 2$. Let $h(t)=Q_4(t)=t^2+1\in \mathbb{F}_q[t]$ and $M\in M_{2\times2}(\mathbb{F}_q)$ such that $P_M(t)=h(t)$. Let $a_1,a_2$ be any PPs over $\mathbb{F}_q$ and $\tau((x_1, x_2)^T)=(a_1(x_1),a_2(x_2))^T$ for any $x_1,x_2\in \mathbb{F}_q$. Let $\sigma=\tau\circ\sigma_M\circ\tau^{-1}$. Then $\sigma$ is a $4$-regular PP over $\mathbb{F}_{q^2}$ by Proposition \ref{proposition 3.8}. Moreover,
\begin{enumerate}
  \item if $a_1,a_2$ are additive, then $\sigma$ is a $4$-regular CPP over $\mathbb{F}_{q^2}$;
  \item if $a_1=a_2=a$ satisfy $a(-x)=-a(x)$ for any $x\in\mathbb{F}_q$ and
  \begin{equation*}
  M=M(h(t))=\left(
  \begin{array}{ccc}
    0 & 1\\
    -1 & 0\\
    \end{array}
  \right),
  \end{equation*}
  then $\sigma$ is a $4$-regular CPP over $\mathbb{F}_{q^2}$;
  \item if $a_2=e$ and
  \begin{equation*}
  M=\left(
  \begin{array}{ccc}
    -1 & m\\
    -2m^{-1} & 1\\
    \end{array}
  \right)
  \end{equation*}
  where $m\in \mathbb{F}_q^*$, then $\sigma$ is a $4$-regular CPP over $\mathbb{F}_{q^2}$.
\end{enumerate}
\end{Proposition}
\begin{proof}
\begin{enumerate}
  \item Since $a_1,a_2$ are additive and $h(-1)\neq 0$, by Theorem \ref{Theorem 3.2}(3), we have $\sigma+e$ is a PP over $\mathbb{F}_{q^2}$, and then $\sigma$ is a 4-regular CPP over $\mathbb{F}_{q^2}$.
  \item Let
  \begin{equation*}
  M_1=\left(
  \begin{array}{ccc}
    1 & 1\\
    -1 & 1\\
    \end{array}
  \right).
  \end{equation*}
  Then $\sigma_{M_1}$ is a PP over $\mathbb{F}_{q^2}$. By calculating, we have
  \begin{equation*}
  \begin{aligned}
  (\sigma+e)\begin{pmatrix}
  x_1\\
  x_2\\
  \end{pmatrix}
  =
  \begin{pmatrix}
  x_1+x_2\\
  -x_1+x_2\\
  \end{pmatrix}
  =
  \sigma_{M_1}\begin{pmatrix}
  x_1\\
  x_2\\
  \end{pmatrix}
  \end{aligned}.
  \end{equation*}
  So $\sigma+e=\sigma_{M_1}$ is a PP over $\mathbb{F}_{q^2}$ and $\sigma$ is a 4-regular CPP over $\mathbb{F}_{q^2}$.
  \item Let \begin{equation}
  \begin{aligned}
  \tau_1\begin{pmatrix}
  x_1\\
  x_2\\
  \end{pmatrix}
  =
  \begin{pmatrix}
  a_1(x_1)+a_1(x_2)\\
  2m^{-1}x_1\\
  \end{pmatrix}
  ,
  \tau_2\begin{pmatrix}
  x_1\\
  x_2\\
  \end{pmatrix}
  =
  \begin{pmatrix}
  2^{-1}mx_2\\
  a_1^{-1}(x_1)-2^{-1}mx_2\\
  \end{pmatrix}
  \end{aligned}
  \end{equation}
  and
  \begin{equation*}
  M_1=\left(
  \begin{array}{ccc}
    1 & -1\\
    1 & 1\\
    \end{array}
  \right).
  \end{equation*}
  It is easy that $\tau_1,\tau_2,\sigma_{M_1}$ are PPs over $\mathbb{F}_{q^2}$. By calculating, we have
  \begin{equation}
  \begin{aligned}
  (\sigma+e)\begin{pmatrix}
  x_1\\
  x_2\\
  \end{pmatrix}
  =
  \begin{pmatrix}
  a_1(-a_1^{-1}(x_1)+mx_2)+x_1\\
  -2m^{-1}a_1^{-1}(x_1)+2x_2\\
  \end{pmatrix}
  =
  \tau_1\circ\sigma_{M_1}\circ\tau_2\begin{pmatrix}
  x_1\\
  x_2\\
  \end{pmatrix}
  \end{aligned}.
  \end{equation}
  So $\sigma+e=\tau_1\circ\sigma_{M_1}\circ\tau_2$ is a PP over $\mathbb{F}_{q^2}$ and $\sigma$ is a 4-regular CPP over $\mathbb{F}_{q^2}$.
\end{enumerate}
\qed\end{proof}

\begin{Remark}
\begin{enumerate}
  \item In Proposition \ref{proposition 4.2}(2), $M=M(h(t))$. In this case, we need one more hypothesis on $a_1, a_2$ that  $a_1=a_2=a$ satisfy $a(-x)=-a(x)$ for any $x\in\mathbb{F}_q$.     We can not find an example with less restricted about $a_1$ and $a_2$.
  \item In Proposition \ref{proposition 4.2}(3), in general, $\tau_1\circ\tau_2\neq e$. So we can not use Theorem \ref{Theorem 3.3}(2) to get the cycle structure of $\sigma+e=\tau_1\circ\sigma_{M_1}\circ\tau_2$.
\end{enumerate}
\end{Remark}




\subsection{$r=5$}
In this subsection, we present constructions of 5-regular CPPs.

\begin{Proposition}\label{proposition 4.3}
Assume that $p\neq5$. Let $h(t)=Q_5(t)=\frac{t^5-1}{t-1}=t^4+t^3+t^2+t+1\in\mathbb{F}_q[t]$ and $M=M(h(t))\in M_{4\times4}(\mathbb{F}_q)$. Let $a$ be any PP over $\mathbb{F}_q$  and $\tau((x_1,x_2,x_3,x_4)^T)=(a(x_1),x_2,x_3,x_4)^T$ for any $x_i\in\mathbb{F}_q$. Let $\sigma=\tau\circ\sigma_M\circ\tau^{-1}$. Then $\sigma$ is a $5$-regular CPP over $\mathbb{F}_{q^4}$.
\end{Proposition}
\begin{proof}
By Proposition \ref{proposition 3.7}, $\sigma$ is a 5-regular PP over $\mathbb{F}_{q^4}$. We only need to prove that $\sigma+e$ is a permutation of $\mathbb{F}_{q^4}$. Indeed, for any $y=(y_1,y_2,y_3,y_4)^T\in\mathbb{F}_{q^4}$, the equation $(\sigma+e)(x)=y$ yields the following system of equations
\begin{equation}
\left\{
\begin{aligned}
x_1+a(x_2)=y_1\\
x_2+x_3=y_2\\
x_3+x_4=y_3\\
-a^{-1}(x_1)-x_2-x_3=y_4
\end{aligned}
\right.
\end{equation}
By the second and fourth equations, we get $x_1=a(-y_2-y_4)$. Then it is easy to find that the system has only one solution. Hence $\sigma+e$ is a PP over $\mathbb{F}_{q^4}$ and $\sigma$ is a 5-regular CPP over $\mathbb{F}_{q^4}$.

\qed\end{proof}

\begin{Remark}
\begin{enumerate}
  \item If $M$ is any $4\times4$ matrix over $\mathbb{F}_q$ such that $P_M(t)=h(t)$, we have not find a general method to construct $5$-regular CPPs over $\mathbb{F}_{q^4}$. We will consider this in further research.
  \item If $q^2\equiv1 \mod 5$, then there exists a quadratic polynomial $h(t)\in \mathbb{F}_q[t]$ such that $h(t)\ |\ Q_5(t)$. By Section \ref{Section 3}, we can construct $5$-regular CPPs over $\mathbb{F}_{q^2}$. We will study this explicitly in further research.
\end{enumerate}
\end{Remark}

\subsection{$r=6$}\label{subsection4.4}
In this subsection, we present constructions of 6-regular CPPs.

\begin{Proposition}\label{proposition 4.4}
Assume that $p\neq2, 3.$ Let $h(t)=Q_6(t)=t^2-t+1\in\mathbb{F}_q[t]$ and $M\in M_{2\times2}(\mathbb{F}_q)$ such that $P_M(t)=h(t)$. Let $a_1,\ a_2$ be any PPs over $\mathbb{F}_q$ and $\tau((x_1,x_2)^T)=(a_1(x_1),a_2(x_2))^T$ for any $x_1,\ x_2\in \mathbb{F}_q$. Let $\sigma=\tau\circ\sigma_M\circ\tau^{-1}$. Then $\sigma$ is a $6$-regular PP over $\mathbb{F}_{q^2}$ by Proposition \ref{proposition 3.8}. Moreover,
\begin{enumerate}
  \item if $a_1,\ a_2$ are additive, then $\sigma$ is a $6$-regular CPP over $\mathbb{F}_{q^2}$;
  \item if $a_2=e$ and
  \begin{equation*}
  M=M(h(t))=\left(
  \begin{array}{ccc}
    0 & 1\\
    -1 & 1\\
    \end{array}
  \right),
  \end{equation*}
  then $\sigma$ is a $6$-regular CPP over $\mathbb{F}_{q^2}$;
  \item if $a_2=e$ and
  \begin{equation*}
  M=\left(
  \begin{array}{ccc}
    -1 & m\\
    -3m^{-1} & 2\\
    \end{array}
  \right)
  \end{equation*}
  where $m \in\mathbb{F}_q^*,$ then $\sigma$ is a $6$-regular CPP over $\mathbb{F}_{q^2}$.
\end{enumerate}
\end{Proposition}
\begin{proof}
\begin{enumerate}
  \item Since $a_1,\ a_2$ are additive and $h(-1)\neq0$, by Theorem \ref{Theorem 3.2}(3) we have $\sigma+e$ is a PP over $\mathbb{F}_{q^2}$, and then $\sigma$ is a 6-regular CPP over $\mathbb{F}_{q^2}$.
  \item Let \begin{equation}
  \begin{aligned}
  \tau_1\begin{pmatrix}
  x_1\\
  x_2\\
  \end{pmatrix}
  =
  \begin{pmatrix}
  a_1(x_1)+a_1(x_2)\\
  -x_1\\
  \end{pmatrix}
  ,
  \tau_2\begin{pmatrix}
  x_1\\
  x_2\\
  \end{pmatrix}
  =
  \begin{pmatrix}
  -x_2\\
  a_1^{-1}(x_1)+x_2\\
  \end{pmatrix}
  \end{aligned}
  \end{equation}
  and
  \begin{equation*}
  M_1=\left(
  \begin{array}{ccc}
    1 & 1\\
    -1 & 0\\
    \end{array}
  \right).
  \end{equation*}
  It is easy that $\tau_1,\tau_2,\sigma_{M_1}$ are PPs over $\mathbb{F}_{q^2}$. By calculating, we have
  \begin{equation}
  \begin{aligned}
  (\sigma+e)\begin{pmatrix}
  x_1\\
  x_2\\
  \end{pmatrix}
  =
  \begin{pmatrix}
  x_1+a_1(x_2)\\
  -a_1^{-1}(x_1)\\
  \end{pmatrix}
  =
  \tau_1\circ\sigma_{M_1}\circ\tau_2\begin{pmatrix}
  x_1\\
  x_2\\
  \end{pmatrix}
  \end{aligned}.
  \end{equation}
   So $\sigma+e=\tau_1\circ\sigma_{M_1}\circ\tau_2$ is a PP over $\mathbb{F}_{q^2}$ and $\sigma$ is a 6-regular CPP over $\mathbb{F}_{q^2}$.
   \item Let \begin{equation}
  \begin{aligned}
  \tau_1\begin{pmatrix}
  x_1\\
  x_2\\
  \end{pmatrix}
  =
  \begin{pmatrix}
  a_1(x_1)+a_1(x_2)\\
  3m^{-1}x_1\\
  \end{pmatrix}
  ,
  \tau_2\begin{pmatrix}
  x_1\\
  x_2\\
  \end{pmatrix}
  =
  \begin{pmatrix}
  3^{-1}mx_2\\
  a_1^{-1}(x_1)-3^{-1}mx_2\\
  \end{pmatrix}
  \end{aligned}
  \end{equation}
  and
  \begin{equation*}
  M_1=\left(
  \begin{array}{ccc}
    2 & -1\\
    1 & 1\\
    \end{array}
  \right).
  \end{equation*}
  It is easy that $\tau_1,\tau_2,\sigma_{M_1}$ are PPs over $\mathbb{F}_{q^2}$. By calculating, we have
  \begin{equation}
  \begin{aligned}
  (\sigma+e)\begin{pmatrix}
  x_1\\
  x_2\\
  \end{pmatrix}
  =
  \begin{pmatrix}
  a_1(-a_1^{-1}(x_1)+mx_2)+x_1\\
  -3m^{-1}a_1^{-1}(x_1)+3x_2\\
  \end{pmatrix}
  =
  \tau_1\circ\sigma_{M_1}\circ\tau_2\begin{pmatrix}
  x_1\\
  x_2\\
  \end{pmatrix}
  \end{aligned}.
  \end{equation}
   So $\sigma+e=\tau_1\circ\sigma_{M_1}\circ\tau_2$ is a PP over $\mathbb{F}_{q^2}$ and $\sigma$ is a 6-regular CPP over $\mathbb{F}_{q^2}$.
\end{enumerate}
\qed\end{proof}

\begin{Remark}
In Proposition \ref{proposition 4.4}(2)(3), in general, $\tau_1\circ\tau_2\neq e$. So we can not use Theorem \ref{Theorem 3.3}(2) to get the cycle structure of $\sigma+e=\tau_1\circ\sigma_{M_1}\circ\tau_2$.
\end{Remark}

\subsection{$r=7$}
In this subsection, we present constructions of 7-regular CPPs.

\begin{Proposition}\label{proposition 4.5}
Assume that $p\neq7$. Let $h(t)=Q_7(t)=\frac{t^7-1}{t-1}=t^6+t^5+t^4+t^3+t^2+t+1\in\mathbb{F}_q[t]$ and $M=M(h(t))\in M_{6\times6}(\mathbb{F}_q)$. Let $a$ be any PP over $\mathbb{F}_q$  and $\tau((x_1,x_2,\ldots,x_6)^T)=(a(x_1),x_2,\ldots,x_6)^T$ for any $x_i\in\mathbb{F}_q$. Let $\sigma=\tau\circ\sigma_M\circ\tau^{-1}$. Then $\sigma$ is a $7$-regular CPP over $\mathbb{F}_{q^6}$.
\end{Proposition}
\begin{proof}
By Proposition \ref{proposition 3.7}, $\sigma$ is a 7-regular PP over $\mathbb{F}_{q^6}$. We only need to prove that $\sigma+e$ is a permutation of $\mathbb{F}_{q^6}$. Indeed, for any $y=(y_1,y_2,\ldots,y_6)^T\in\mathbb{F}_{q^6}$, the equation $(\sigma+e)(x)=y$ yields the following system of equations
\begin{equation}
\left\{
\begin{aligned}
x_1+a(x_2)=y_1\\
x_2+x_3=y_2\\
x_3+x_4=y_3\\
x_4+x_5=y_4\\
x_5+x_6=y_5\\
-a^{-1}(x_1)-x_2-x_3-x_4-x_5=y_6
\end{aligned}
\right.
\end{equation}
By the second, fourth, sixth equations, we get $x_1=a(-y_2-y_4-y_6)$. Then it is easy to find that the system has only one solution. Hence $\sigma+e$ is a PP over $\mathbb{F}_{q^6}$ and $\sigma$ is a 7-regular CPP over $\mathbb{F}_{q^6}$.

\qed\end{proof}

\begin{Remark}
\begin{enumerate}
  \item If $M$ is any $6\times6$ matrix over $\mathbb{F}_q$ such that $P_M(t)=h(t)$, we have not find a general method to construct $7$-regular CPPs over $\mathbb{F}_{q^6}$. We will consider this in further research.
  \item If $q^2\equiv1\mod 7$, then there exists a quadratic polynomial $h(t)\in \mathbb{F}_q[t]$ such that $h(t)\ |\ Q_7(t)$. By Section \ref{Section 3}, we can construct $7$-regular CPPs over $\mathbb{F}_{q^2}$. We will study this explicitly in further research.
  \item If $q^3\equiv1\mod 7$, then there exists a cubic polynomial $h(t)\in \mathbb{F}_q[t]$ such that $h(t)\ |\ Q_7(t)$. By Section \ref{Section 3}, we can construct $7$-regular CPPs over $\mathbb{F}_{q^3}$. Next, we consider the cases $p=2$. While for other cases, we  will study them in further research.
\end{enumerate}
\end{Remark}

Now we present constructions of 7-regular CPPs over $\mathbb{F}_{q^3}$ where $p=2$. First, we consider the additive cases.

\begin{Proposition}\label{proposition 4.6}
Assume that $p=2$. Let $h(t)=t^3+t^2+1$ and $M\in M_{3\times3}(\mathbb{F}_q)$ such that $P_M(t)=h(t)$. Let $\tau$ be any additive PP over $\mathbb{F}_{q^3}$. Let $\sigma=\tau\circ\sigma_M\circ\tau^{-1}$. Then $\sigma$ and $\sigma+e$ are both $7$-regular CPPs over $\mathbb{F}_{q^3}$.
\end{Proposition}
\begin{proof}
Since $p=2$, we have
$$Q_7(t)=t^6+t^5+t^4+t^3+t^2+t+1$$
 $$\ \ \ =(t^3+t^2+1)(t^3+t+1)$$
 and
 $$(t+1)^7-1=t^7+7t^6+21t^5+35t^4+35t^3+21t^2+7t+1-1$$
$$\ \ \ =t^7+t^6+t^5+t^4+t^3+t^2+t=tQ_7(t).$$
Then $h(t)=t^3+t^2+1\ |\ (t^7-1)$ and $h(t)\ |\ ((t+1)^7-1)$. By Theorem \ref{Theorem 3.2}(2)(4), we have $\sigma$ is a 7-regular CPP over $\mathbb{F}_{q^3}$ and $\sigma+e$ is a 7-regular PP over $\mathbb{F}_{q^3}$. Meanwhile, $(\sigma+e)+e=\sigma$ is a PP over $\mathbb{F}_{q^3}$. So $\sigma+e$ is also a 7-regular CPP over $\mathbb{F}_{q^3}$.

\qed\end{proof}

By the same method used in Proposition \ref{proposition 4.6}, we can prove
\begin{Proposition}
Assume that $p=2$. Let $h(t)=t^3+t+1$ and $M\in M_{3\times3}(\mathbb{F}_q)$ such that $P_M(t)=h(t)$. Let $\tau$ be any additive PP over $\mathbb{F}_{q^3}$. Let $\sigma=\tau\circ\sigma_M\circ\tau^{-1}$. Then $\sigma$ and $\sigma+e$ are both $7$-regular CPPs over $\mathbb{F}_{q^3}$.
\end{Proposition}

Now we consider the general cases.

\begin{Proposition}\label{proposition 4.8}
Assume that $p=2$. Let $h(t)=t^3+t^2+1$ and $M\in M_{3\times3}(\mathbb{F}_q)$ such that $P_M(t)=h(t)$. Let $a_1,\ a_2$ be any PPs over $\mathbb{F}_q$ and $\tau_1((x_1,x_2,x_3)^T)=(x_1,a_1(x_2),x_3)^T$, $\tau_2((x_1,x_2,x_3)^T)=(x_1,a_2(x_2),x_3)^T$. Let $\sigma=\tau_1\circ\sigma_M\circ\tau_2$. Then
\begin{enumerate}
  \item if
  \begin{equation*}
  M=M(h(t))=\left(
  \begin{array}{ccc}
    0 & 1 & 0 \\
    0 & 0 & 1 \\
    1 & 0 & 1 \\
  \end{array}
  \right),
  \end{equation*}
  then $\sigma$ is a CPP over $\mathbb{F}_{q^3}$. Moreover, if $a_1\circ a_2=e$, then $\sigma$ is a $7$-regular CPP over $\mathbb{F}_{q^3}$;
  \item if
  \begin{equation*}
  M=\left(
  \begin{array}{ccc}
    0 & 1 & 1 \\
    1 & 0 & 0 \\
    1 & 0 & 1 \\
  \end{array}
  \right),
  \end{equation*}
  then $\sigma$ is a CPP over $\mathbb{F}_{q^3}$. Moreover, if $a_1\circ a_2=e$, then $\sigma$ is a $7$-regular CPP over $\mathbb{F}_{q^3}$.
\end{enumerate}
\end{Proposition}

\begin{proof}
\begin{enumerate}
  \item \label{1}In this case,
    \begin{equation}
  \begin{aligned}
  (\sigma+e)\begin{pmatrix}
  x_1\\
  x_2\\
  x_3\\
  \end{pmatrix}
  =
  \begin{pmatrix}
  x_1+a_2(x_2)\\
  x_2+a_1(x_3)\\
  x_1\\
  \end{pmatrix}
  \end{aligned}.
  \end{equation}
  It is easy to see that $\sigma+e$ is a PP over $\mathbb{F}_{q^3}$ and $\sigma$ is a CPP over $\mathbb{F}_{q^3}$. Moreover, when $a_1\circ a_2=e$, we have $\tau_1\circ\tau_2=e$. Combining $h(t)\ |\ (t^7-1)$, by Theorem \ref{Theorem 3.3}(2), we have $\sigma$ is a 7-regular CPP over $\mathbb{F}_{q^3}$.
  \item In this case,
  \begin{equation}
  \begin{aligned}
  (\sigma+e)\begin{pmatrix}
  x_1\\
  x_2\\
  x_3\\
  \end{pmatrix}
  =
  \begin{pmatrix}
  x_1+a_2(x_2)+x_3\\
  a_1(x_1)+x_2\\
  x_1\\
  \end{pmatrix}
  \end{aligned}.
  \end{equation}
  It is easy to see that $\sigma+e$ is a PP over $\mathbb{F}_{q^3}$ and $\sigma$ is a CPP over $\mathbb{F}_{q^3}$. When $a_1\circ a_2=e$, it is similar to (\ref{1}) that $\sigma$ is a 7-regular CPP over $\mathbb{F}_{q^3}$.
\end{enumerate}
\qed\end{proof}

By the same method used in Proposition \ref{proposition 4.8}, we can prove

\begin{Proposition}
Assume that $p=2$. Let $h(t)=t^3+t+1$ and $M\in M_{3\times3}(\mathbb{F}_q)$ such that $P_M(t)=h(t)$. Let $a_1,\ a_2$ be any PPs over $\mathbb{F}_q$ and $\tau_1((x_1,x_2,x_3)^T)=(x_1,a_1(x_2),x_3)^T$, $\tau_2((x_1,x_2,x_3)^T)=(x_1,a_2(x_2),x_3)^T$. Let $\sigma=\tau_1\circ\sigma_M\circ\tau_2$. Then
\begin{enumerate}
  \item if
  \begin{equation*}
  M=M(h(t))=\left(
  \begin{array}{ccc}
    0 & 1 & 0 \\
    0 & 0 & 1 \\
    1 & 1 & 0 \\
  \end{array}
  \right),
  \end{equation*}
  then $\sigma$ is a CPP over $\mathbb{F}_{q^3}$. Moreover, if $a_1\circ a_2=e$, then $\sigma$ is a $7$-regular CPP over $\mathbb{F}_{q^3}$;
  \item if
  \begin{equation*}
  M=\left(
  \begin{array}{ccc}
    1 & 1 & 1 \\
    1 & 0 & 0 \\
    1 & 0 & 1 \\
  \end{array}
  \right),
  \end{equation*}
  then $\sigma$ is a CPP over $\mathbb{F}_{q^3}$. Moreover, if $a_1\circ a_2=e$, then $\sigma$ is a $7$-regular CPP over $\mathbb{F}_{q^3}$.
\end{enumerate}
\end{Proposition}

\begin{Remark}
\begin{enumerate}
  \item In fact, Proposition \ref{proposition 4.8}(2) is the same as \cite[Theorem 2]{XZZ2022}. We give a different and simpler proof.
  \item If $\tau$ is additive, then $\sigma$ and $\sigma+e$ are both $7$-regular CPPs over $\mathbb{F}_{q^3}$ by Proposition \ref{proposition 4.6}. While if $\tau$ is not additive, then $\sigma$ is a $7$-regular CPP over $\mathbb{F}_{q^3}$ in some cases by Proposition \ref{proposition 4.8}, but $\sigma+e$ is not $7$-regular in general by calculating.
\end{enumerate}
\end{Remark}

\subsection{$r$ is an   odd positive integer}\label{subsection 4.6}
In this subsection, we present constructions of $r$-regular CPPs for arbitrary odd positive integer $r$.

The following proposition is the generalization of \cite[Theorem 1]{XZZ2022}.
\begin{Proposition}\label{proposition 4.10}
Assume that $r\geq3$ is an odd positive integer which is relatively prime to $p$. Let $h(t)=\frac{t^r-1}{t-1}\in \mathbb{F}_q[t]$, $d=\mathrm{deg}(h(t))=r-1$ and $M=M(h(t))\in M_{d\times d}(\mathbb{F}_q)$. Let $a_1,\ a_2$ be any PPs over $\mathbb{F}_q$ and
$$\tau_1((x_1,x_2,\ldots,x_d)^T)=(a_1(x_1),x_2,\ldots,x_d)^T,$$
$$\tau_2((x_1,x_2,\ldots,x_d)^T)=(a_2(x_1),x_2,\ldots,x_d)^T.$$
Let $\sigma=\tau_1\circ\sigma_M\circ\tau_2$. Then
\begin{enumerate}
  \item $\sigma$ is a CPP over $\mathbb{F}_{q^d}$;
  \item If $a_1\circ a_2=e$ and $r$ is a prime, then $\sigma$ is an $r$-regular CPP over $\mathbb{F}_{q^d}$;
  \item If $a_1\circ a_2=e$ and $r$ is a composite number, then $\sigma$ is not an $r$-regular CPP over $\mathbb{F}_{q^d}$.
\end{enumerate}
\end{Proposition}
\begin{proof}
\begin{enumerate}
  \item It is trivial that $\sigma$ is a PP over $\mathbb{F}_{q^d}$. For $\sigma+e$, we have for any $y=(y_1,y_2,\ldots,y_d)^T\in\mathbb{F}_{q^d}$, the equation $(\sigma+e)(x)=y$ yields the following system of equations
  \begin{equation}
  \left\{
  \begin{aligned}
  x_1+a_1(x_2)=y_1\\
  x_2+x_3=y_2\\
  \vdots\\
  x_{d-1}+x_d=y_{d-1}\\
  -a_2(x_1)-x_2-x_3-\cdots-x_{d-1}=y_d
  \end{aligned}
  \right.
  \end{equation}
  Then
   $$-a_2(x_1)=(-a_2(x_1)-x_2-\cdots -x_{d-1})+\sum\limits^{\frac{d-1}{2}}_{j=1}(x_{2j}+x_{2j+1}) =y_d+\sum\limits^{\frac{d-1}{2}}_{j=1}{y_{2j}}$$
  and
  $$x_1=a_2^{-1}(-(y_d+\sum\limits^{\frac{d-1}{2}}_{j=1}{y_{2j}})).$$
  Now it is easy to find that the system has only one solution. Hence $\sigma+e$ is a PP over $\mathbb{F}_{q^d}$ and $\sigma$ is a CPP over $\mathbb{F}_{q^d}.$
  \item Since $a_1\circ a_2=e$, we have $\tau_1\circ\tau_2=e$. Since $r$ is an odd prime, by Proposition \ref{proposition 3.7}, $\sigma$ is an $r$-regular PP over $\mathbb{F}_{q^d}$. Combining (1), we have $\sigma$ is an $r$-regular CPP over $\mathbb{F}_{q^d}.$
  \item Since $r$ is a composite number, it is easy to see that $h(t)=\frac{t^r-1}{t-1}$ is reducible, $(h(t), \frac{t^r-1}{Q_r(t)})\neq1$ and $h(-1)\neq0$. Moreover, $M=M(h(t))$, by Proposition \ref{proposition 3.9}, we have that $\sigma$ is not $r$-regular.
\end{enumerate}
\qed\end{proof}

\begin{Remark}
\begin{enumerate}
  \item When $p=2$, Proposition \ref{proposition 4.10} becomes \cite[Theorem 1]{XZZ2022}. Our proof is different from theirs and is much easier than theirs.
  \item When $r$ is an even positive integer, the situation is complicated. In some cases, $\sigma$ is a CPP over $\mathbb{F}_{q^d}$; while in other cases, $\sigma$ is not a CPP over $\mathbb{F}_{q^d}$. For example, when $r=6$, $\sigma$ is a CPP over $\mathbb{F}_{q^d}$ by Proposition \ref{proposition 4.4}(2). But when $r=4$, $\sigma$ is a not CPP over $\mathbb{F}_{q^d}$, see Proposition \ref{proposition 4.2}(2) and a simple discussion.
  \item If $M$ is any $d\times d$ matrix over $\mathbb{F}_q$ such that $P_M(t)=h(t)$, we have not find a general method to construct regular CPPs over $\mathbb{F}_{q^d}$. We will consider this in further research.
\end{enumerate}
\end{Remark}

\section {\bf{Concluding remarks}}\label{section Concluding remarks}

This paper considered the $r$-regular  complete permutation property of maps with the form $f=\tau\circ\sigma_M\circ\tau^{-1}$ and give a general construction of regular PPs and regular CPPs over extension fields.

Theorem \ref{Theorem 3.3} give a general construction of $r$-regular PPs for any positive integer $r$, see Proposition \ref{proposition 3.7} and  Proposition \ref{proposition 3.8}.
When $\tau$ is  additive, Theorem \ref{Theorem 3.2} give a general construction of $r$-regular CPPs for any positive integer $r$, see Proposition \ref{proposition 3.4} and  Proposition \ref{proposition 3.5}.
When $\tau$ is not additive, Section \ref{section 4} give many examples of regular CPPs over the extension fields for $r=3,4,5,6,7$ and for arbitrary odd positive integer $r$.

By the examples in Section \ref{section 4}, we find that, for any given positive integer $r$, in order to get  $r$-regular CPPs  over $\mathbb{F}_{q^d}$ with the form $f=\tau\circ\sigma_M\circ\tau^{-1}$, it is easy to get polynomials $h(t)$ satisfy suitable conditions. The difficulty is that find   suitable  matrices $M$ and suitable PPs $\tau$ to make sure $f+e$ is a PP over $\mathbb{F}_{q^d}$.  We will consider this in further research.


 {}

\begin{thebibliography}{}

 \bibitem{A1969} Ahmad S.: Cycle structure of automorphisms of finite cyclic groups. J. Comb. Theory 6, 370-374 (1969).
 \bibitem{B2003} Biryukov A.: Analysis of involutional ciphers: Khazad and Anubis. Fast Softw. Encryption 2887, 45-53 (2003).
\bibitem{CR2015} Canteaut A., Roue J.: On the behaviors of affine equivalent S-boxes regarding differential and linear attacks, in: Advances in Cryptology - EUROCRYPT 2015 - 34th Annual International Conference on the Theory and Applications of Cryptographic Techniques, Sofia, Bulgaria, April 26-30, 2015, in: Lecture Notes in Computer Science, Part I, vol. 9056, Springer, pp. 45-74 (2015).
\bibitem{CCZ1998} Carlet C., Charpin P.,  Zinoviev V.:  Codes, bent functions and permutations suitable for DES-like cryptosystems. Designs, Codes, Cryptogr., vol. 15, no. 2, 125-156 (1998).
\bibitem{CMS2016} Charpin P., Mesnager S., Sarkar S.: Involutions over the Galois field $\mathbb{F}_{2^n}$. IEEE Trans. Inf. Theory 62 (4), 2266-2276  (2016).
\bibitem{CWZ2021}Chen Y., Wang L., Zhu S.: On the constructions of n-cycle permutations. Finite Fields Appl. 73, 101847 (2021).
\bibitem{CM2018} Coulter R.S., Mesnager S.: Bent functions from involutions over $\mathbb{F}_{2^n}$. IEEE Trans. Inf. Theory 64 (4), 2979-2986 (2018).
\bibitem{DP2013}  Dempwolff U., Muller P.: Permutation polynomials and translation planes of even order. Adv. Geometry, vol. 13, no. 2, 293-313 (2013).
\bibitem{DL2008} Diffie W., Ledin G. (translators): SMS4 encryption algorithm for wireless networks. https://eprint.iacr.org/2008/329.pdf.
\bibitem{D2013} Ding C.: Cyclic codes from some monomials and trinomials. SIAM J. Discrete Math., vol. 27, no. 4, 1977-1994 (2013).
\bibitem{DQWYY2015} Ding C., Qu L., Wang Q.,  Yuan J., Yuan P.: Permutation trinomials over finite fields with even characteristic. SIAM J. Discrete Math., vol. 29, no. 1, 79-92 (2015).
\bibitem{DY2006}  Ding C.,  Yuan J.: A family of skew Hadamard difference sets. J. Combinat. Theory A, vol. 113, no. 7, 1526-1535, (2006).
\bibitem{D1999a} Dobbertin H.: Almost perfect nonlinear power functions on GF($2^n$): The Niho case. Inf. Comput., vol. 151, nos. 1-2, 57-72 (1999).
\bibitem{D1999b} Dobbertin H.: Almost perfect nonlinear power functions on GF($2^n$): The Welch case. IEEE Trans. Inf. Theory, vol. 45, no. 4, 1271-1275 (1999).
\bibitem{FFZ2011} Feng D., Feng X., Zhang W., et al.: Loiss: a byte-oriented stream cipher. In: IWCC'11 Proceedings of the Third International Conference on Coding and Cryptology,   109-125. Springer, New York (2011).
\bibitem{F1982} Fredricksen H.: A survey of full length nonlinear shift register cycle algorithms. SIAM Rev. 24(2), 195-221 (1982).
\bibitem{G1962} Gallager R.: Low-density parity-check codes. IRE Trans. Inf. Theory 8 (1),  21-28  (1962).
\bibitem{G1967} Golomb S.W.: Shift Register Sequences. Holden-Day Inc, Laguna Hills (1967).
\bibitem{GG2005} Golomb S.W., GongG.:Signal Design for Good Correlation. For Wireless Communication, Cryptography, and Radar. Cambridge University Press, New York (2005).
\bibitem{H2015} Hou X.D.: Determination of a type of permutation trinomials over finite fields, II. Finite Fields Their Appl., vol. 35, 16-35 (2015).
\bibitem{L2002} Lang, S.: Algebra. Springer New York, (2002).
\bibitem{LM1991} Lidl R., Mullen G.L.: Cycle structure of Dickson permutation polynomials. Math. J. Okayama Univ. 33, 1-11 (1991).
\bibitem{LM1984} Lidl R.,  Muller W. B.: Permutation polynomials in RSA-cryptosystems. in Advances in Cryptology. Boston, MA, USA: Springer,  293-301 (1984).
\bibitem{M1942}Mann H.B.: The construction of orthogonal Latin squares. Ann. Math. Stat. 13(4), 418-423 (1942).
\bibitem{MM2009}Markovski S., Mileva A.: Generating huge quasigroups from small non-linear bijections via extended Feistel function. Quasigroups Relat. Syst. 17(1), 91-106 (2009).
\bibitem{M1973}McFarland R. L.: A family of difference sets in non-cyclic groups. J. Combinat. Theory A, vol. 15, no. 1,  1-10 (1973).
\bibitem{M2016} Mesnager S.: On constructions of bent functions from involutions, in: 2016 IEEE International Symposium on Information Theory (ISIT), IEEE,  110-114 (2016).
\bibitem{MM2012a}Mileva A., Markovski S.: Quasigroup representation of some Feistel and generalized Feistel ciphers. In: ICT Innovations 2012. Advances in Intelligent Systems and Computing, vol. 207,  161-171. Springer, Berlin (2012).
\bibitem{MM2012b}Mileva A., Markovski S.: Shapeless quasigroups derived by Feistel orthomorphisms. Glas. Mat. 47(67), 333-349 (2012).
\bibitem{M1995}Mittenthal L.: Block substitutions using orthomorphic mappings. Adv. Appl. Math. 16(10), 59-71 (1995).
\bibitem{M1997}Mittenthal L.: Nonlinear dynamic substitution devices and methods for block substitutions employing coset decompositions and direct geometric generation. US Patent 5647001 (1997).
\bibitem{MP2014}Muratovic-Ribic A., Pasalic E.: A note on complete polynomials over finite fields and their applications in cryptography. Finite Fields Appl. 25, 306-315 (2014).
\bibitem{M2021}Muratovic-Ribic, A., On generalized strong complete mappings and mutually orthogonal Latin squares. Ars Mathematica Contemporanea 21(2) (2021).
\bibitem{NR1982}Niederreiter H., Robinson K.H.: Complete mappings of finite fields. J. Aust. Math. Soc. A 33(2), 197-212 (1982).
\bibitem{RC2004}Rubio I., Corrada C.: Cyclic decomposition of permutations of finite fields obtained using monomials,. Finite Fields and Applications, LNCS 2948,   254-261, Springer, New York (2004).
\bibitem{RMCC2008}Rubio I., Mullen G.L., Corrada C., Castro F.N.: Dickson permutation polynomials that decompose in cycles of the same length. Contemp. Math. 461, 229-240 (2008).
\bibitem{R2003}Rudolf Lidl, Harald Niederreiter: Finite fields. Encyclopedia of Mathematics and ITS Applications, (2003).
\bibitem{SSP2012}Sakzad A., Sadeghi M.R., Panario D.: Cycle structure of permutation functions over finite fields and their applications. Adv. Math. Commun. 6(3), 347-361 (2012).
\bibitem{SV1995}Schnorr C.P., Vaudenay S.: Black box cryptanalysis of hash networks based on multipermutations. In: Advances in Cryptology-Eurocrypt'94,  47-57. Springer, New York (1995).
\bibitem{SG2012}Stanica P., Gangopadhyay S., Chaturvedi A., Gangopadhyay A.K., Maitra S.: Investigations on bent and negabent functions via the negaHadamard transform. IEEE Trans. Inf. Theory 58, 4064-4072 (2012).
\bibitem{V1994}Vaudenay S.: On the need for multipermutations: cryptanalysis of MD4 and SAFER. In: Fast Software Encryption-FSE'94. Lect. Notes Comput. Sci., vol. 1008,  286-297. Springer, New York (1994).
\bibitem{V1999}Vaudenay S.: On the Lai-Massey scheme. In: Advances in Cryptology-ASIACRYPT-99. Lect. Notes Comput. Sci., vol. 1716, 8-19. Springer, New York (1999).
\bibitem{XZZ2022}Xu, X., Zeng, X., Zhang, S: Regular complete permutation polynomials over $\mathbb{F}_{2^n}$. Des. Codes Cryptogr. 90, 545-575 (2022).
\bibitem{ZHC2015}Zha Z., Hu  L., Cao X.: Constructing permutations and complete permutations over finite fields via subfield-valued polynomials. Finite Fields Appl. 31  162-177  (2015).













\end{thebibliography}
\end{document}